\documentclass{llncs}
\usepackage{latexsym}

\begin{document}
\title{Bounded Counter Languages%
\thanks{Research partially supported by ``Deutsche Akademie
der Naturforscher Leopoldina'', grant number BMBF-LPD~9901/8-1
of ``Bundesministerium f{\" u}r Bildung und For\-schung''.} 
}
\author{Holger Petersen\\ 
           Reinsburgstr.\ 75\\ D-70197 Stuttgart}
\institute{}
\maketitle
\begin{abstract}
We show that deterministic finite automata equipped with $k$ two-way heads are 
equivalent to deterministic machines  with a single two-way input head and 
$k-1$ linearly  bounded counters if the accepted language is strictly bounded, i.e., 
a subset of  $a_1^*a_2^*\cdots a_m^*$ for a fixed sequence of symbols $a_1, a_2,\ldots,
a_m$. Then we investigate linear speed-up for counter machines. Lower and upper time
bounds for concrete recognition problems are shown, implying that in general linear speed-up 
does not hold for counter machines. For bounded languages we develop a technique for 
speeding up computations by any constant factor at the expense of adding a fixed number
of counters.
\end{abstract}

\section{Introduction}
The computational model investigated in
the present work is the two-way counter machine as defined in \cite{FMR68}. 
Recently, the power of this model
has been compared to quantum automata and probabilistic automata \cite{YKI05,SY11}. 

We will show that bounded counters and heads are 
equally powerful for finite deterministic devices, provided the languages 
under consideration are strictly bounded.
By equally powerful we mean that up to a single head each two-way
input head of a deterministic finite machine can be simulated 
by a counter bounded by the input length and vice versa. 
The condition that the input is bounded cannot be removed in general, since it is known
that deterministic finite two-way two-head automata are more powerful
than deterministic two-way one counter machines 
if the input is not bounded \cite{Duris82}.
The special case of equivalence between deterministic
one counter machines and two-head automata over a single letter
alphabet has been shown with the help of a two-dimensional
automata model as an intermediate step in \cite{Morita77jap}, see also 
\cite{Morita77}.

Adding  resources to a computational model should intuitively increase its power. 
This is true in the case of time and space 
hierarchies for Turing machines, see the text book \cite[Chapter~9]{Sipser06}.  
The language from \cite{Duris82} cited above is easily acceptable with two counters and thus shows that two counters are 
more powerful than one. A further growing 
number of unbounded counters does however not increase the power of these machines due to the classical result of Minsky \cite{Minsky61},
showing that  machines with two counters are universal. 
Thus the formally defined hierarchy of language classes accepted by  machines
with a growing number of counters collapses to the second level.
A tight hierarchy is obtained if the counters are linearly bounded \cite{Monien80} or if the machines are working in real-time 
\cite[Theorem~1.3]{FMR68}. In the latter case the machines are allowed to make one step per input symbol  and there is obviously 
no difference in accepting power between one-way and two-way access to the input. If we restrict the input to be read one-way
(sometimes called on-line \cite{Greibach76}), a hierarchy in exponential time can be shown \cite{Petersen09}.

The starting point of our investigation of time hierarchies is Theorem~1.3 of \cite{FMR68}, where
the authors show that the language of marked binary palindromes has time complexity  $\Theta(n^2/\log n)$ on two-way 
counter machines.
By techniques from descriptional complexity \cite{LV93} for the lower bounds  
we are able to separate classes of machines with different numbers of counters.
A main motivation for this work is of course a fundamental interest in
the way the capabilities of a computational device influence its
power. There are however other more technical consequences.  
 Special cases of strictly bounded languages are languages over a single letter
alphabet. Our simulation of heads by counters thus eliminates the need for hierarchy
results separating $k+1$ bounded counters from $k$ bounded counters for deterministic
devices with the help of single letter alphabet languages; the
hierarchy for deterministic devices stated in Theorem~3 of
\cite{Monien80} follows from the corresponding result for multi-head
automata in Theorem~1 of \cite{Monien80}. 

In comparison with
multi-head automata, counter machines appear to be affected
by slight technical changes of the
definition. There are, e.g., several natural ways to define counter
machines with counters bounded by the input length \cite{Ritchie72}.
One could simply require that the counters never overflow, with the
drawback that this property is undecidable in general.  Alternatively the
machine could block in the case of overflow, acceptance then being
based on the current state, a specific error condition could be
signaled to the machine, with the counter being void, or the counter
could simply remain unchanged.  The latter model is clearly equivalent
to the known concept of a {\em simple} multi-head automaton. All these
variants, which seem to be slightly different in power, can easily be
simulated with the help of heads and therefore coincide, at least for strictly
bounded input.

Regarding time bounded computations, we present an  algorithm for the recognition of 
marked palindromes working with only two counters, while the upper bound outlined in the proof 
of \cite[Theorem~1.3]{FMR68} requires at least three counters (one for storing $\log_2 m$ and two for
encoding portions of the input).
We show that counter machines lack general linear speed-up. Other models of computation with this property include Turing machines with tree storages
\cite{Huehne93} and Turing machines with a fixed
alphabet and a fixed number of tapes \cite{BCS11}. By adapting the witness language, we disprove a claim that these machines satisfy speed-up for
 polynomial time bounds \cite{Greibach76}.
Finally we  present a class of languages where linear speed-up can be achieved by adding a fixed number of counters.

\section{Definitions}

A language is {\em bounded} if it is a subset of 
$w_1^*w_2^*\cdots w_m^*$ for a fixed sequence of
words $w_1, w_2\ldots,  w_m$ (which are not necessarily distinct).
We call $w_1^*w_2^*\cdots w_m^*$ the {\em bound} of the language.
A language is {\em strictly bounded} if it is a subset of 
$a_1^*a_2^*\cdots a_m^*$ for distinct symbols $a_1, a_2, \ldots, a_m$.
A maximal sequence of symbols
$a_i$ in the input will be called a {\em block}.

Formal definitions of variants of counter machines can be found in \cite{FMR68}. 
We only point out a few essential features
of  the models. 
The main models of computation investigated here are the {\em $k$-head
  automaton}, the {\em bounded counter machine with $k$ counters},
and the {\em register machine with $k$ bounded registers}, cf.\ 
\cite{Monien80}.  The two former types of automata have read-only input tapes
bordered by end-markers. Therefore on an
input of length $n$ there are $n+2$ different positions that can be read.
The $k$-head automaton is equipped with $k$
two-way heads that may move independently on the input tape and
transmit the symbols read to the finite control.  Note that informally
two heads of the basic model
cannot see each other, i.e., the machine has no way of finding out that
they happen to be reading the same input square. Heads  with the capability 
``see'' each others are called sensing.

The bounded counter machine is equipped with a single two-way head
and $k$ counters that can count up to the input length. The operations
it can perform on the counters are increment, decrement, and zero-test.
The machines start their operation with all heads next to the left
end-marker  (on the first input symbol if the input is not empty) and 
all counters set to zero.  Acceptance is by final state and can occur with
the input head at any position. 

The register machine receives an input number in its first register,
all other registers are initially zero. We compare register machines 
with the other types of machines by identifying nonnegative integers
and strings over a single letter alphabet of a corresponding length.
All machines have deterministic and nondeterministic variants and accept
by entering a final state.

The set of {\em marked palindromes} over a binary alphabet is
$L = \{ x\$ x^T \mid x\in\{ 0, 1\}^*\},$
where $x^T$ denotes the reversal of $x$. 

With the restriction that at least approximately half of the input is filled by zeros, we obtain
$$L' = \{ x0^{|x|}\$ 0^{|x|}x^T \mid x\in\{ 0, 1\}^*\},$$
and a further restriction to only a logarithmic information content leads to the family
$$L_m = \{ x0^{2^{|x|/m}-|x|}\$ 0^{2^{|x|/m}-|x|}x^T \mid x\in\{ 0, 1\}^*\}$$
for $m\ge 1$.

\section{Equivalence of Heads and Counters for Deterministic Machines}\label{secequivalence}

In order to simplify the presentation we assume below
without loss of generality that a multi-head
automaton moves exactly one head in every step.  Suppose a multi-head
automaton operates on a word from a strictly bounded language. We call
a step in which a head passes from one block of identical symbols to a
neighboring block or an end-marker an {\em event}.  The head moved in this step is said
to {\em cause the event}.

\begin{lemma}\label{segment}
  Let the input of a deterministic multi-head automaton be strictly bounded.
  It is possible to determine whether a head will cause the next event
  (under the assumption that no other head does) by inspecting an
  input segment of fixed length around
  this head independently of the input size. If it can cause the next event it 
  will also be determined 
  at which boundary of the scanned block it will happen.
\end{lemma}
\begin{proof}
 Let the automaton have $r$ internal states.
If the head moves to a square at least $r$ positions away from
its initial position within the same block before the next event, 
then some state must have
been repeated (since all other heads keep reading the same symbols)
and the automaton, receiving the same information from
its heads while no event occurs, 
will continue to work in a cycle until a head causes an
event. Therefore it suffices to simulate the machine on a segment of
$2r-1$ symbols under and around the head under consideration
(or less symbols if the segment exceeds
the boundaries of the input tape). 
We assume that no other head causes the  next event, therefore
at most $2r^2-r$
different partial configurations consisting of state,
symbols read by the heads and the position on the segment are possible 
before one of the following happens:
\begin{itemize}
\item The head leaves the current block causing an event.
\item The head leaves the segment 
 around the initial position of the head.
\item The automaton gets into a loop repeating partial configurations within the segment.
\end{itemize}
In the the two former cases the boundary at which the next event 
is possibly caused by the head under cosideration is clear. In the latter case the head cannot
cause the next event. 
\qed\end{proof}

\begin{theorem}\label{strict-equiv}
  Deterministic multi-head automata
  with $k$ two-way heads and deterministic bounded counter machines with
  $k-1$ counters are equivalent over strictly bounded languages.
\end{theorem}
\begin{proof}
 Any bounded counter machine with
$k-1$ counters can easily be simulated by a $k$-head automaton, independently
of the structure of the input string. The multi-head automaton
simulates the input head of the counter machine with one of its
heads and encodes the values stored by the counters as the distances of
the remaining head positions from the left end-marker.

For the converse direction we will describe the simulation of a
deterministic multi-head automaton $M$ by a deterministic 
bounded counter machine
$C$ and call $C$'s single input head its pointer, thus avoiding some
ambiguities. 
The counters and the pointer of $C$ are assigned to the heads of $M$,
this assignment varies during the simulation.

The counters will store distances between head positions
and boundaries of blocks of the input, where distances to 
left or right boundaries may occur in the course of the 
simulation. The 
finite amount of information consisting of the assignment and 
the type of distance for each counter is stored in $C$'s finite 
control. Depending on the type of distance stored for a head,
movements of $M$'s heads are translated into the corresponding
increment and decrement operations. If a distance to a left 
boundary is stored, a left movement causes a decrement and
a right movement an increment operation on the counter. 
For distances to a 
right boundary the operations are reversed. 

We divide the computation of the multi-head automaton $M$ into
intervals.  Each interval starts with a configuration in which at
least one head is next to a boundary (i.e., on one of the two positions
left or right of the boundary between blocks), 
one of these heads being represented by
$C$'s pointer. Notice that the initial configuration satisfies this
requirement.  Each interval except the last one ends when the next
event occurs. After this event the machine is again in a
configuration suitable for a new interval.

Counter machine $C$ always
updates the symbols read by the heads of $M$ (initially the first
input symbol or the right end-marker) and keeps this information in its 
finite control. A counter assigned to a
head encodes the number of symbols within the block 
that are to the left resp.\ right of the head position. The
counter machine also maintains the information which of these two
numbers is stored.  

We start our description of the algorithm that $C$ executes in a
configuration with the property that at least one head of $M$ is next
to a boundary. One of these heads is represented by the single
pointer of $C$.  First $C$ moves its pointer to every block that is
read by some head of $M$. It can determine these blocks uniquely from
the symbols stored in the finite control. It moves its pointer next to
the boundary that is indicated by the type of distance stored on the
counter assigned to the head under consideration. While the counter is
not zero it decrements the counter and moves the pointer towards the
head position. Then it determines whether this head could cause the
next event, provided that no other head does, by applying
Lemma~\ref{segment}. For this purpose the pointer reads the
surrounding segment of length $2r-1$ without losing the head position.  
Then the test is carried out in $C$'s finite control.

Suppose the head could cause the next event at some boundary. Then $C$
updates the contents of the counter to reflect the distance of
$M$'s head from that boundary. If the block has the form $a_i^{x_i}$ and
the head is on position $n\ge 1$ in this block, then the counter machine
moves its pointer towards the boundary  where the event possibly occurs and
measures the distance with the help of the counter. Thus $C$ updates
the contents of the counter with either $n-1$ or $x_i-n$, respectively, depending
on whether the event can occur at the left or right boundary.
If no event can be caused by the head currently considered, one of the distances is stored,
say to the left boundary.

These operations are carried out for every head. Finally $C$ moves its
pointer back to the initial position, which is possible since it is next
to a boundary. Then it starts to simulate $M$ step by step, translating 
head movements into counter operations according to the distance 
represented by the counter contents. If $M$ gets into
an accepting state, $C$ accepts. If the pointer leaves its block 
the next interval starts. If a counter
is about to be decremented from zero this operation is not carried out. 
Instead the current pointer position is recorded in this counter and 
the roles of pointer and counter are interchanged.  
The symbols read by the heads that 
pointer and counter are now assigned to, as well as the
internal state of $M$ are updated. This information can clearly be
kept in $C$'s finite control. A new interval starts.

The initial configuration of $C$ has all counters set to zero with the pointer
and all simulated heads reading the first input symbol. The assignment of heads to
counters is arbitrary, all counters store the distance to the left  boundary.
\qed\end{proof}

The equivalence of heads and counters implies that the intermediate
concept of simple heads --- two-way heads that cannot distinguish
different input symbols --- also coincides in power with counters over
strictly bounded languages. It is open whether the analogous
equivalence holds over arbitrary input or,
as conjectured in \cite{Morita77jap,Morita77}, 
simple heads are more powerful than counters.%
 \footnote{The $n$-bounded
  counters of \cite{Morita77jap,Sugata77jap} can count from 0 up to
  $n$ and can be tested for these values. They are easily seen to be
  equivalent to simple heads. The class of deterministic two-way machines
  equipped with $k$ such counters is denoted $C(k)$ in
  \cite{Morita77jap,Sugata77jap}, while machines with $k$ unbounded counters 
  are denoted by $D(k)$.}  
Finding a candidate for the separation even of one simple head
from a counter
seems to be difficult. In recognizing the language $L_4 = \{ ww \mid
w\in\{ 0, 1\}^*\}$ the full power of a simple head was used in 
\cite{Sugata77jap}, 
but this is not necessary. A counter machine can
first check that the input contains an even number of symbols. Starting
with the first symbol it then stores the distance of the current symbol
to the left end-marker 
on the counter. Sweeping its head over the entire input it increments
the counter in every second step and thus computes the offset of the 
corresponding symbol, moves its head to this position and compares the
symbols. Then it reverses the computation to return to the initial position 
and moves its head to the next symbol. If all corresponding symbols are equal 
it accepts.

Finally we compare the power of register machines and multi-head automata over
a single letter alphabet. For the simulation we will identify lengths of input
strings and input numbers.
In Lemma~5 of \cite{Monien80} a simulation of $k$-head automata over
a single letter alphabet by $k+1$ register machines is given. The simulation
is rather specialized, since it applies only to subsets of words
that have a length which is a power of two. We will generalize this simulation
to arbitrary languages over a single letter alphabet.

\begin{theorem}
Every deterministic (nondeterministic) $k$-head automaton over
a single letter alphabet can be simulated by a deterministic 
(nondeterministic) $k+1$ register machine. The heads are even allowed
to see each others.
\end{theorem}
\begin{proof}
 First we normalize multi-head automata that can detect heads
scanning the same square such that the heads appear in a fixed
left-to-right sequence on the tape (if some heads are on the same
square we allow any sequence, which includes this fixed one). 
This is easily achieved because the
automata can internally switch the roles of two heads which are about
to be transposed.

The register machine simulating a $k$-head automaton with the help of $k+1$
registers stores in its registers the distances between neighboring heads
or the end-marker, where the distance is the number of steps 
to the right a head would have to carry out in order to 
reach the next head or end-marker. Register~1 represents the distance
of the last head to the right end-marker.
Whenever a head moves the register machine updates the two related distances.
A small technical problem is the distance to the left end-marker, which 
formally should  be $-1$ in the initial configuration. 
The register machine stores the information whether
the left-most heads scan the left end-marker in its finite control. 
In this way 
all distances can be bounded by the input length.\qed\end{proof}

\section{Time-Bounds for Counter Machines}\label{sechierarchy}

The purpose of this section is to establish lower and upper time-bounds on counter machines for concrete recognition problems.

\begin{theorem}\label{unboundedlower}
The recognition of the language $L'$ of marked palindromes with desert requires
$$\frac{n^2-16n\log_2n-dn}{8(\log_2n+2\log_2s+1)}$$ 
steps for input strings of length $n$ on $1$-counter machines and
$$\frac{n^2-16n\log_2n-dn}{8(2k\log_2n+\log_2s)}$$ 
steps on $k$-counter machines in the worst case and for
$n$ sufficiently large, where $s$ and $d$ are constants depending on the specific counter machine.
\end{theorem}
\begin{proof}
Let $M$ be a $1$-counter machine with $s\ge 2$ states accepting $L'$ and let $x$ be an incompressible string with $|x| = m \ge 1$. Consider the accepting
computation of $M$ on $x0^{|x|}\$ 0^{|x|}x^T$ 
and choose position $i$ adjacent to or within  the central portion $0^{|x|}\$ 0^{|x|}$ with a
crossing sequence $c$ having $\ell$ entries of minimum length. Notice that for $1$-counter machines the counter is bounded 
from above by $s(n+2)< 2sn$, since $n\ge 2m+1$.

String $x$ can be reconstructed from the following information:
\begin{itemize}
\item A description of $M$ ($O(1)$ bits).
\item A self-delimiting encoding of the length of $x$ ($2\log_2n$ bits).
\item Position $i$ of $c$ ($\log_2n$ bits).
\item Length $\ell$ of $c$ ($\log_2(4s^2n) $ bits).
\item Crossing sequence $c$ recording the counter contents and the state $M$ enters when crossing position $i$
($\ell(\log_2n+2\log_2s+1)$ bits). 
\item A formalized description of the reconstruction procedure outlined below ($O(1)$ bits).
\end{itemize}

For reconstructing $x$ from the above data, a simulator sets up a section of length $|x|$ followed by all symbols ($0$ or $\$ $)
up to position $i$. Then the simulator cycles through all strings $y$ of length $|x|$ and simulates $M$ step by step. Whenever position
$i$ is reached, it is checked that the current entry of the crossing sequence matches state and counter contents. If not, the current
$y$ is discarded and the next string is set up. If it matches, the simulation continues from state and counter contents of the next 
entry of the crossing sequence. If $M$ accepts, the encoded $x$ has been found and the simulation terminates.

Since $x$ is incompressible, for some constant $d$ compensating the $O(1)$ contributions we must have:
$$ |x| = (n-1)/4 \le \ell(\log_2n+2\log_2s+1) + 4\log_2n + d/4 - 1/4$$
and thus
$$ \ell \ge \frac{n-16\log_2n-d}{4(\log_2n+2\log_2s+1)}.$$

There are $(n-1)/2+2\ge n/2$ positions of crossing sequences with length at least $\ell$, thus we get
$$T(n) \ge \frac{n^2-16n\log_2n-dn}{8(\log_2n+2\log_2s+1)}.$$

For machines with $k\ge 2$ counters we bound the counter contents by the coarse upper bound $n^2$ 
(since the asymptotical bound grows more slowly, this bound suffices). This increases the bound on the length
of the encoding of crossing sequences to $\ell(2k\log_2n+\log_2s)$ bits. The time bound becomes:
$$T(n) \ge \frac{n^2-16n\log_2n-dn}{8(2k\log_2n+\log_2s)}.$$
\qed\end{proof}

Next we present an upper bound for the full language $L$ and give an algorithm that uses
only two counters in comparison to at least three in \cite{FMR68}. We conjecture that this cannot be reduced to one 
counter for subquadratic algorithms.
\begin{theorem}\label{upperone}
Language $L$ of marked palindromes can be accepted 
in $O(n^2/\log n)$ steps by a two-counter machine.
\end{theorem}
\begin{proof} We describe informally the work of a machine $M$ accepting $L$ on an input of length $n\ge 1$. The idea is to encode segments
of length $\log_2 n$ and iteratively compare segment by segment.

First $M$ scans the input and counts the symbols before the $\$ $ (if no  $\$ $ 
is found, $M$ rejects). After the $\$ $ the counter is decremented and the input is rejected, 
if zero is not reached on the right end-marker or a second $\$ $ is encountered. The first scan takes 
$n$ steps if $M$ starts on the leftmost input-symbol as defined in \cite{FMR68}.

First $M$ puts 1 on counter 1, repeatedly reads a symbol, doubles the 
counter contents (exchanging roles for each bit read), and adds 1 if the symbol read was 1. Notice that after such a doubling one of the counter contents is zero. 
Then $M$ makes excursions to the left and to the right counting up on the empty counter and down on the counter holding the encoding until the latter counter
becomes empty or $\$ $ resp.\ an end-marker is reached. The net effect is an (attempted) subtraction of $n/2$ from the encoding.
If the counter becomes zero, the initial encoding was less than $n/2$. 
After each of the excursions, the other counter
is used for returning to the initial position and the encoding is restored. 
If the encoding exceeds $n/2$, the process stops and the segment is compared to the corresponding portion to the right
of $\$ $.
Using the empty counter, $M$ moves to the corresponding portion
and in a symmetrical way as for the encoding decodes the segment. Since the encoding has a 1 as its most significant bit, no excursions are necessary.
In order to return to the last position in the segment, $M$ repeats the encoding process.
Then it continues with the next segment. The iterations stop if the marker $\$ $ is reached.

For the time analysis we omit constant and linear contributions to the 
number of steps, these are accounted for by an appropriately chosen constant factor of the leading term. 
The initial scan is clearly linear. By the doubling procedure the amortized cost of encoding and decoding is linear per segment. Notice that
the excusions are aborted if the counter holding the encoding is empty and thus also the excursions have linear complexity per segment.
Since the number of segments is $O(n/\log n)$, the bound follows.
\qed\end{proof}

\begin{theorem}
For  $k$-counter machines with fixed $k\ge 2$ accepting $L$  in $O(n^2/\log n)$ steps there is no linear speed-up.
\end{theorem}
\begin{proof} The lower bounds for a subset of $L$ in Theorem~\ref{unboundedlower} carry over to the full language. These absolute 
lower bounds show that the algorithm from Theorem~\ref{upperone} cannot be accelerated by an arbitrary factor.
\qed\end{proof}

\begin{remark}
By adding counters and encoding larger segments of the input a speed-up for the recognition of $L$ is possible.
\end{remark}

We now adapt the witness language to the bounds considered in \cite[p.~273]{Greibach76}. There speed-up 
results for deterministic
and nondeterminstic two-way machines with $r$ counters and time bounds of the form $pn^k$ with $p > 1$ and $k\ge 1$ are stated. 
No formal proofs are given, but the preceding section contains a reference to \cite{FMR68}. We note here
that the speed-up results in \cite[Section~5.2]{FMR68} are based on Theorem~1.2, which appears before the definition of
two-way machines and therefore applies to one-way models only.

In the special case of  linear time bounds we will disprove the claimed speed-up. For at least quadratic bounds the 
technique does not apply, since the type of languages considered can be accepted in quadratic time comparing bit by bit and 
constant speed-up by forming blocks of constant size can clearly be achieved. Whether a general 
linear speed-up is possible is open, since there seems to be no efficient way to compress the contents of the input tape.

\begin{theorem}\label{loglower}
The recognition of the language $L_m$ requires
$(m/13)n$ steps on 4-counter machines 
in the worst case and for $n$ sufficiently large.
\end{theorem}
\begin{proof} The proof is adapted from the one given for Theorem~\ref{unboundedlower} and we only describe the differences.

We assume $|x| > m\log_2|x|$ in the following. 
Since the part $x$ is now only a logarithmic portion of the input, $2\log_2\log_2n$ bits suffice for encoding the length. The time bound
is the linear function $pn$, therefore the crossing sequence can be described in $\ell(4(\log_2n+\log_2p)+\log_2s)$ bits. 

For an incompressible $x$ we obtain:
\begin{eqnarray*}
m\log_2n - 2m \le  |x|  & = & m\log_2((n-1)/2) \\
                                   & \le &\ell(4(\log_2n+\log_2p)+\log_2s) + 2\log_2n + d \\
                                   & = & \ell(6\log_2n) + d'
\end{eqnarray*}
with constants $d, d'$ and
$ \ell \ge m/6 - o(1).$
Since the desert is at least $n/2$ symbols long for $|x| > m\log_2|x|$ we get the time bound
$T(n)  \ge (m/12)n - o(n).$
\qed\end{proof}
 
We now have to show that there is a linear recognition algorithm for $L_m$.
\begin{theorem}\label{logupper}
The language $L_m$ can be accepted 
in $(2m+3)n + o(n) $ steps by a 4-counter machine.
\end{theorem}
\begin{proof} In a first left-to-right scan, 
an acceptor $M_m$ for $L_m$ determines the length of the part of the input before
the $\$ $, compares it to the part after $\$ $ and at the same time counts that length again in another counter.
By repeatedly dividing by two, it computes $|x|/m$ in at most $n$ steps 
and checks in a right-to left scan, 
whether the middle portion of the input contains only 0's. While doing so, $M_m$ preserves
$|x|/m$ and computes  $2^{|x|/m}-|x|$ for checking the left half of the input.

Then $M_m$ starts to encode blocks of $|x|/2m$ bits on two counters while using the other two counters for
checking the length of the block. Since $|x| \le m\log_2n$, the encoding can be done in $O(\sqrt{n})$ steps. 
Then $M_m$ moves its input head to the other half of the input, 
decodes the stored information, and contiues with the next iteration. The number of $2m$ iterations of $n + o(n)$
steps each can be counted in the finite control. 
\qed\end{proof}

As an example take $L_{39}$, which can be accepted in $82n$ steps by Theorem~\ref{logupper} for sufficiently large $n$, but a linear speed-up to
$2n$ is impossible by Theorem~\ref{loglower}.

\section{Speed-Up for Counter Machines on Bounded Input}\label{secspeedup}
In Section~\ref{sechierarchy} we have shown that linear speed-up for certain recognition problems
on counter machines can only be achieved by adding counters. 
If the input is compressible, the situation is different and a fixed number of counters depending on
the structure of the language suffices for speeding up by any constant factor.

We will apply the classical result of Fine and Wilf on periodicity of sequences: 
\begin{theorem}[\cite{Fine65}]\label{finewilf}
Let $(f_n)_{n\ge 0}$ and $(g_n)_{n\ge 0}$ be two periodic sequences of
period $h$ and $k$, respectively. If $f_n = g_n$ for $h + k - \mbox{gcd}(h, k)$
consecutive integers, then $f_n = g_n$ for all $n$.
\end{theorem}

\begin{theorem}\label{speedupbounded}
 For every $k$-counter machine
 accepting a bounded language with $m$ blocks
 and operating in time $t(n)$ there is an equivalent counter machine with
 $k+m$ counters operating in time $n + c t(n)$ for any constant $c > 0$.
\end{theorem}
\begin{proof}
The strategy is to encode the input of a counter machine $M$ with
 $k$ counters on $m$ counters of a simulator $M'$ in a first stage and simulate the two-way
machine $M$ with speed-up using the method suitable for one-way machines 
from the proof of Theorem~5.3 in \cite{FMR68}. 

Let the accepted language be a subset of  $w_1^*w_2^*\cdots w_m^*$. Note that
 in general the exponents $k_1, \ldots, k_m$ of $w_1, \ldots, w_m$ for a given input
are not easily recognizable, since the borders between the factors might not be evident. 
We will show that an encoding of the input is nevertheless possible.

The encoding
of the input will work in 
stages. At the end of stage $i$ the encoding covers at least the first $i$ blocks. Let 
$\mu = \max \{ |w_j| \mid 1 \le j \le m\}$. At the start of stage $i$ simulator $M'$ 
reads 
$2\mu$ additional input symbols if possible (if the end of the input is reached, the
suffix is recorded in the finite control). Machine  $M'$  records this string $y$ in the finite control. 
Then for each conjugate 
$vu$ of $w_i = uv$ the input segment $y$ is compared to a prefix of length $2\mu$ of
$(vu)^{\omega}$. If no match is found, the segment $y$ and its position is stored in the finite control
and the next stage starts.
If a match to some $vu$ is found, $M'$ assigns a counter to this 
$vu$, stores the number of copies of $vu$ within $y$ (including a trailing prefix of $vu$ 
if necessary) are recorded in the finite control, and then $M'$ continues to count
the number of  copies of $vu$ found in the input until the end of the input is reached or
the next $|vu|$ symbols do not match $vu$. These $|vu|$ symbols are stored in the 
finite control. The process ends when the input is
exhausted. Then $M'$ has stored an encoding of the input, which can be recovered by 
concatenating the segments stored in the finite control and the copies of the conjugates of
the $w_i$.

We now argue that at the end of stage $i$ the encoding has reached the end of block $i$. Initially
the claim holds vacuously. Suppose by induction that the claim holds for stage $i-1$.  The next 
$2\mu$ symbols are beyond block $i-1$ and if they do not match a power of a conjugate, 
then the string is not embedded into a block. Thus it extends over the end of block $i$ and the 
claim holds. Otherwise a match between some $vu$ and $y$ is found. String $y$ is a 
factor of some $w_j^\omega$ and by the Fine and Wilf
Result (Theorem~\ref{finewilf}) both are powers of the same $z$ since 
$2\mu \ge h + k - \mbox{gcd}(h, k)$ with $h = |vu|$ and $k = |w_j|$. Therefore $M'$ is able to
encode all of the copies of $w_j$ on the counter. Notice that $vu$ is not necessarily a conjugate of
$w_j$. Thus the claim also holds in this case.

Since each stage requires at most one counter, the $m$ additional counters suffice.

After encoding the input, the two-way input-head of $M$ is simulated 
by $M'$ with the help of $m$ counters and its input-head, which is 
used as an initially empty counter measuring the 
distance to the right end-marker. Whenever the
simulated head enters a block $i$, the simulator starts to decrement the corresponding 
counter and increment a counter available (since it has just been zero). The input
head is simulated on $w_i$, where the position of the input head 
position modulo $|w_i|$ is kept in the finite control of $M'$. Also the assignment of
counters to blocks is dynamic and stored in the finite control.

Now $M'$ is replaced by $M''$ with compressed counter contents and operating with speed-up according to the
proof of Theorem~5.3 from \cite{FMR68}. \qed\end{proof}

\end{document}